\documentclass{ws-procs975x65}

\def\beq{\begin{equation}}
\def\eeq{\end{equation}}

\begin{document}

\title{Critical line of the Ising model on 2-dimensional CDT and its dual}
\author{George M. Napolitano$^{a,\dag}$ and Tatyana S. Turova$^{a,b,\ddag}$}

\address{$^{a}$Centre for Mathematical Sciences, Lund University \\
S\"olvegatan 18, 22100, Lund, Sweden\\
$^\dag$E-mail: gmn@maths.lth.se\\
$^{\ddag}$E-mail: tatyana@maths.lth.se}

\address{$^{b}$IMPB, Russian Academy of Sciences \\ Pushchino, Moscow Region, Russia}

\begin{abstract}
In this paper we study the annealed coupling of an Ising model with 2-dimensional causal dynamical triangulation model. After a short review of previous results, we prove the existence of the so-called critical line and derive its analytical properties and asymptotics. In addition, relations between the model and its dual are investigated.
\end{abstract}

\keywords{Causal dynamical triangulation; Ising model; critical line; dual.}

\bodymatter


\section{Introduction}

The causal dynamical triangulation (CDT) model was first introduced  in Ref. \refcite{AL98} as a non-perturbative theory of quantum gravity. The basic idea behind this model is to interpret the gravitational path-integral  by discretizing spacetime manifolds through so-called \textit{causal triangulations}. Thus, the problem of computing the gravitational path-integral is replaced by a statistical mechanical problem, that is the study of the partition function of an ensemble of finite random triangulations, whose weight is given by the Boltzmann factor $e^{- S_\mu[T]}$, where $S_\mu[T]$ is the discrete Einstein-Hilbert action. Having described the partition function, one can study its ``thermodynamic'' limit, when the size of triangulations grows to infinity. We refer the reader to Ref. \refcite{ADJ-qg}, where the subject is treated in detail.



Even in two dimensions, the CDT model turns out to possess a rich mathematical structure. In particular, it has been proved (see Refs. \refcite{MYZ01,DJW3}) that it is critical at a certain (known) value $\mu_c$ of the cosmological constant. This critical behaviour is reflected in the geometrical properties (Hausdorff and spectral dimensions) of the infinite triangulations: they show 2-dimensional features at criticality, whereas in the subcritical regime they are essentially 1-dimensional. 

Matter is inserted in the theory by running  a statistical mechanical model on the random triangulation. Technically, this can be implemented by weighting the Boltzmann factor associated to each triangulation $T$ with the partition function of an Ising model running  on $T$. This is generally called \textit{annealed} coupling (see Ref. \refcite{ADJ-qg} for further details). 

%

As in the pure gravity case, the coupled model is believed to become critical when the parameters of the model (the cosmological constant $\mu$ and the inverse temperature $\beta$ of the spin system) are suitably tuned. 
However, the exact value of the critical parameters are still unknown and only bounds on the region where they lie have been provided (see Refs. \refcite{HSYZ} and \refcite{NT15}).


In this paper we review some recent results, first appeared in Ref. \refcite{NT15}, on the annealed coupling between the Ising model and the random causal triangulations of the plane. In particular, we show the existence of a critical line in the positive $(\beta,\mu)$ quarter plane,  give a description of some of its analytical properties and provide bounds on the region where the critical line is located. Furthermore, we add some new results on the relation between the critical line and its dual.

\section{Ising model on 2-dimensional CDT}

\begin{definition}\label{DM}
A \textit{causal triangulation} $T$ is a rooted planar locally finite connected graph satisfying the following properties.
\begin{enumerate}
	\item The set of vertices at graph distance $i$ from the root, together with the edges connecting them, form a cycle, denoted by $S_i=S_i(T)$ and called $i$-th \textit{slice} (when there is only one vertex the corresponding cycle has only one edge, i.e. it is a loop).
	\item All faces of the graph are triangles, with the only exception of the external face.
	\item One edge attached to the root vertex is marked, we call it \textit{root edge}. 
\end{enumerate}
\end{definition}

The last condition in the above definition is a technical requirement, needed to cancel out possible rotational symmetries around the root.




%

We denote by $\mathcal{T}_{N,l}$ the set of causal triangulations with $N$ slices and fixed number $l$ of vertices on $S_N$. In the following, we shall denote by $V(T)$, $E(T)$ and $F(T)$, the vertex set, the edge set and the triangles set of $T$, respectively. The cardinality of a set $A$ will be denote by $|A|$.

Given a triangulation $T \in \mathcal{T}_{N,l}$, a spin configuration 
$\sigma$ on $T$ is an element of the set
\begin{equation*}
	\Omega(T) = \{ \sigma \in \{ +1, -1\}^{V(T)} \}.
\label{eq:Omega}
\end{equation*}

The partition function of the coupled system (Ising model and random causal triangulations) is defined as
\begin{equation}\label{Z}
     Z_{N,l}(\beta,\mu) = \sum_{ T \in \mathcal{T}_{N,l}} \sum_{\sigma \in \Omega(T)} e^{-\beta H(T, \sigma)-\mu |F(T)|},
\end{equation}
where $\beta \geq 0$ is the inverse temperature, $\mu \geq 0$ is the cosmological constant and
\begin{equation}
	H(T,\sigma) = - \sum_{(u,v) \in E(T)} \sigma_u \sigma_v
\label{eq:ham}
\end{equation} 
is the classical nearest-neighbour spin interaction Hamiltonian.

\section{The critical line.}

Since the set $\mathcal{T}_{N,l}$ is countable, the series defined in (\ref{Z}) might diverge. On the other hand, for any fixed $\beta \geq 0$, each term of the series is a positive and strictly decreasing function of $\mu$, therefore there exists a value $\mu^{c}_{N,l}(\beta)$, such that $Z_{N,l}(\beta,\mu) < \infty$, for any $\mu > \mu^{c}_{N,l}(\beta)$, and $Z_{N,l}(\beta,\mu) = \infty$, for any $\mu < \mu^{c}_{N,l}(\beta)$.

It is easy to see that the sequence $\mu^{c}_{N,l}(\beta)$ is bounded from above. Indeed, from the inequality
\begin{equation*}
	H(\sigma,T) \geq - |E(T)|,
\end{equation*}
it follows that
\begin{equation*}
	Z_{N,l}(\beta,\mu) \leq \sum_{ T \in \mathcal{T}_{N,l}} 2^{|V(T)|} e^{\beta|E(T)|-\mu |F(T)|},
\end{equation*}
where the series on the right  is finite (see Ref. \refcite{NT15}) for all $N$ and $l$ whenever
 $\mu > \frac{3}{2} (\log 2 + \beta)$. Therefore
\begin{equation}
	\mu^{c}_{N,l}(\beta) \leq \frac{3}{2} (\log 2 + \beta).
\label{U}
\end{equation}

In the following theorem, we prove the existence of a critical value $\mu_{c}(\beta)$ for each  $\beta\geq 0$, above which the partition function is finite for any $N$ and $l$.


\begin{theorem}\label{Te1}
For each $\beta \geq 0$ and $l \in \mathbb N$, as $N\to\infty$
\begin{equation}
\mu^{c}_{N,l}(\beta) \nearrow \mu_{c}(\beta):=
\lim_{N\to\infty} \mu^{c}_{N,l}(\beta) ,
\label{t1}
\end{equation}
where  the limit exists and does not depend on $l$. Furthermore, the function $ \mu_{c}(\beta)$ is non-decreasing in $\beta$. 
\end{theorem}

\begin{proof}
First we show that for each $\beta \geq 0$ and $l \in \mathbb N$ function $\mu^{c}_{N,l}(\beta)$
is non-decreasing in $N$. 

Consider the partition function in (\ref{Z}) for the triangulations with a fixed last slice $S_N$ with 
$k$ vertices. It is straightforward to derive the following bound
\begin{equation*}\label{Z1}
\begin{split}
     Z_{N,k}(\beta,\mu) & \geq  \sum_l \sum_{ T \in \mathcal{T}_{N-1,l}} \sum_{\sigma \in \Omega(T)} e^{-\beta H(T, \sigma)-\mu |F(T)|}\\
			& \times e^{-\mu (k+l)}\sum_{\sigma' \in \{-1, +1\}^k} e^{-\beta \sum_{(i,j)\in S_N}\sigma'_{i}\sigma'_{j}} \\
			& = \sum_l  e^{-\mu (k+l)}Z_{N-1,l}(\beta,\mu) \left( \sum_{\sigma' \in \{-1, +1\}^k} e^{-\beta \sum_{(i,j)\in S_N}\sigma'_{i}\sigma'_{j}}\right) \\
			& \geq  \left(2e^{-\beta-\mu}\right)^k \sum_l  e^{-\mu l}Z_{N-1,l}(\beta,\mu) .
\end{split}
\end{equation*}

Hence, if $Z_{N-1,l}(\beta,\mu) =\infty$ for some $l$ then also  $Z_{N,k}(\beta,\mu) =\infty$.
This implies for any $k,l,N$
\begin{equation}\label{Z3}
\mu_{N-1,l}^{c}(\beta)\leq \mu_{N,k}^{c}(\beta).
\end{equation}
When $k=l$ this  together with the uniform bound (\ref{U}) 
yields existence of $$\lim_{N\rightarrow\infty} \mu_{N,k}^{c}(\beta)=:\mu_{k}^{c}(\beta),$$ and then also by (\ref{Z3}) for any fixed $k,l$
\[
\mu_{l}^{c}(\beta)\leq \mu_{k}^{c}(\beta),
\]
as well as 
\[
\mu_{k}^{c}(\beta) \leq \mu_{l}^{c}(\beta).
\]
We conclude that for any fixed $k$ there exists $\mu_{c}(\beta)$
such that 
\[\lim_{N\rightarrow\infty} \mu_{N,k}^{c}(\beta)=\mu_{c}(\beta).\]
This proves statement (\ref{t1}).

To prove the monotonicity of $\mu_{c}(\beta)$ consider again $Z_{N,k}(\beta,\mu)$.
Assume, $Z_{N,k}(\beta,\mu)<\infty.$ Then 
\begin{equation*}
   \frac{\partial }{\partial \beta}  Z_{N,l}(\beta,\mu) = 
\sum_{ T \in \mathcal{T}_{N,l}} e^{-\mu |F(T)|}
\sum_{(u,v)\in E(T)}
\sum_{\sigma \in \Omega(T)} \sigma_u\sigma_v e^{-\beta H(T, \sigma)}.
\end{equation*}
But  for any $T \in \mathcal{T}_{N,l}$
 and $(u,v)\in E(T)$
\[\sum_{\sigma \in \Omega(T)} \sigma_u\sigma_v e^{-\beta H(T, \sigma)}\geq 0\]
by the First Griffiths inequality for ferromagnetic spin systems. Hence, 
\[\frac{\partial }{\partial \beta}  Z_{N,l}(\beta,\mu) \geq 0, \]
i.e.,  for any $\beta<\beta'$
\[Z_{N,l}(\beta,\mu) \leq Z_{N,l}(\beta',\mu), \]
which in turn yields 
\[ \mu_{N,l}^{c}(\beta)\leq \mu_{N,l}^{c}(\beta').\]
Then monotonicity of $\mu_{c}(\beta)$ follows from here by the statement (\ref{t1}).
\end{proof}

Since the function $\mu_{c}(\beta)$ is non-negative and bounded from above by a linear function (see  (\ref{U})),
its monotonicity implies that it has at most a countable number of jumps. Therefore 
we shall call the graph of $\mu_{c}(\beta)$, $\beta \geq 0$, \textit{the critical line}. 

Theorem \ref{Te1} implies 
that $\mu_{c}(\beta)= \sup_{N,k} \mu_{N,k}^{c}(\beta)$, and therefore 
by a  theorem  of Ref. \refcite{NT15} proved for $\sup_{N,k} \mu_{N,k}^{c}$, we 
immediately have the following result.

\begin{theorem} \label{thm:cr_bounds}
The critical line $\mu_{c}(\beta)$ is  monotone non-decreasing, it has at most a countable number of jumps and it is contained in the region of the $(\beta,\mu)$-plane defined by
\begin{equation*}
	\Delta = \{ (\beta,\mu) \in \mathbb{R}^2 : l(\beta) \leq \mu \leq \frac{3}{2}\log 2 + \beta+ \frac{1}{2} \log\cosh\beta\},
\end{equation*}
where
\begin{equation*}
	l(\beta) = \max\{\frac{3}{2}\log 2 + \frac{3}{2}\log\cosh\beta, \log 2 +\frac{3}{2} \beta \}.
\end{equation*}
\end{theorem}

Above the critical line, that is in the region
\begin{equation*}
	\Sigma = \{ (\beta,\mu) \in \mathbb{R}^2 : \beta \geq 0, \mu > \mu_{c}(\beta)\},
\end{equation*}
the partition function is finite for all $N$ (and $l$). We call this \textit{subcritical region}. On the other hand, below the critical line, for each $(\beta,\mu)$, $Z_N(\beta,\mu) = \infty$ for $N$ large enough.

Furthermore, from Thm. \ref{thm:cr_bounds} it follows that the asymptotic behaviour of $\mu_{c}(\beta)$ is given by
\begin{gather}
	\mu_{c}(\beta) = \frac{3}{2} \log 2 + O(\beta), \quad \text{ as } \beta \to 0, \\
	\mu_{c}(\beta) = \log 2 + \frac{3}{2} \beta + O(e^{-2\beta}), \quad \text{ as } \beta \to \infty \label{eq:crl_inf}.
\end{gather}

In particular, eq. (\ref{eq:crl_inf}) shows that the ground-state spin configurations are dominant at low temperature, as expected (see Ref. \refcite{HSYZ}).

\section{Duality}
In this section we compare the partition function (\ref{Z}), to the partition function of the system constructed by placing the Ising spins on the faces of a causal triangulation, that is on the vertices of the dual graph $T^*$ of a causal triangulation $T$.


Hence, the partition function of the dual model is defined as
\begin{equation*}
Z_{N,l}^{\Delta}(\beta,\mu) = \sum_{ T \in \mathcal{T}_{N,l}}  \sum_{\sigma \in \Omega(T^*)} e^{- \beta  H(T^*, \sigma)-\mu |F(T)|}.
\end{equation*}

Applying a Kramers-Wannier duality argument, it can be shown (see Ref. \refcite{NT15}) that  
the partition functions of the two systems are related by the equation
\begin{equation}
	Z_{N,l}^{\Delta}(\beta^*,\mu^*) = \left(2\sinh{2\beta} \right)^{ -\frac{l}{4}} 	Z_{N,l}(\beta,\mu),
\label{eq:ZduZ}
\end{equation}
where $\beta^*=\beta^*(\beta)$ and $\mu^*=\mu^*(\mu,\beta)$ are such that
\begin{gather}
	\tanh \beta^* = e^{- 2 \beta} , \label{eq:tempdual} \\
	\mu^* = \mu + \frac{1}{4} \log 2 -\frac{3}{4}\log{(\sinh{2\beta})} \label{eq:cosmdual}.
\end{gather}

Equation (\ref{eq:ZduZ}) implies that the subcritical region of the original model is bijectively mapped onto the subcritical region of the dual model by the transformations (\ref{eq:tempdual})-(\ref{eq:cosmdual}). Hence, in particular the critical lines of the two models are related by the equation
\begin{equation}
	\mu^*_{c}(\beta) = \mu_{c}(\beta) + \frac{1}{4} \log 2 -\frac{3}{4}\log{(\sinh{2\beta})}.
\label{eq:mucrdual_mucr}
\end{equation}   

By the relation (\ref{eq:mucrdual_mucr}), Thm. \ref{thm:cr_bounds} implies the following.

\begin{corollary}
The dual critical line $\mu^*_{c}(\beta^*)$ is  monotone non-decreasing, it has at most a countable number of jumps and it is contained in the region of the $(\beta^*,\mu^*)$-plane defined by

\begin{equation*}
	\Delta^* = \{ (\beta^*,\mu^*) \in \mathbb{R}^2 : l^*(\beta^*) \leq \mu^* \leq 2 \log 2 + \frac{1}{2} \beta^* + \log\cosh\beta^*\},
\end{equation*}
where
\begin{equation*}
	l^*(\beta^*) = \max\{2 \log 2 + \frac{3}{2}\log\cosh\beta^*, \log 2 +\frac{3}{2} \beta^* \}.
\end{equation*}
\end{corollary}

In particular, we have the following asymptotic results
\begin{gather*}
	\mu^*_{c}(\beta^*) = 2 \log 2 + O(\beta^{*}), \quad \text{ as } \beta^* \to 0, \\
	\mu^*_{c}(\beta^*) = \log 2 + \frac{3}{2} \beta^* + O(e^{-2\beta^*}), \quad \text{ as } \beta^* \to \infty.
\end{gather*}

%

\section{Conclusions}
In this paper we presented some results on the annealed coupling of an Ising model and the causal dynamical triangulation model. We proved the existence of a critical line $\mu_c(\beta)$ delimiting the region of the $(\beta,\mu)$-plane where the partition function $Z_{N,l}(\beta,\mu)$ of the system is finite for all $N$ and $l$ and showed that the critical line is a monotone non-decreasing function of $\beta $ and has at most a countable number of jumps. Moreover, the critical line is proved to lie in a region of the plane shrinking to a straight line when $\beta \to \infty$. This gives in turn the exact asymptotic behaviour of the critical line and shows that  ground-state spin configurations are dominant at low temperature. In addition, we obtained exact relations between the partition functions of the model itself and its dual, as well as between their critical lines.



\end{document}